\documentclass[12pt,letterpaper]{article}

\usepackage[margin=1in]{geometry}
\usepackage{setspace}

\usepackage{cite}
\usepackage{amsmath,amssymb,amsfonts}
\usepackage{graphicx}
\usepackage{textcomp}

\usepackage{pgfplots}
\usepackage{tikz}
\usetikzlibrary{arrows.meta,positioning}
\usepackage{pgf}
\usepackage{graphics} 
\usepackage{epsfig} 
\usepackage{mathptmx} 
\usepackage{times} 
\usepackage{color}
\usepackage{float}
\usepackage{ntheorem}
\usepackage{calrsfs}
\DeclareMathAlphabet{\pazocal}{OMS}{zplm}{m}{n}
\usepackage{subcaption}
\usepackage{cleveref}
\usepackage{algorithm,algcompatible}
\algnewcommand\INPUT{\item[\textbf{Input:}]}%
\algnewcommand\OUTPUT{\item[\textbf{Output:}]}%

\newtheorem{theorem}{Theorem}[section]

\newtheorem{lemma}{Lemma}[section]
\newtheorem{prop}{Proposition}[section]

\newtheorem{remark}{Remark}[section]

\makeatletter
\newtheoremstyle{MyNonumberplain}%
{\item[\theorem@headerfont\hskip\labelsep ##1\theorem@separator]}%
{\item[\theorem@headerfont\hskip\labelsep ##3\theorem@separator]}
\makeatother
\theoremstyle{MyNonumberplain}
\theorembodyfont{\upshape}
\newtheorem{refproof}{Proof}
\newtheorem{proof}{Proof}

\DeclareMathOperator{\interior}{int}

\begin{document}
	
\onehalfspacing
	
\title{Distributed Model Predictive Control with Reconfigurable Terminal Ingredients for Reference Tracking\footnote{This work is supported by the Swiss Innovation Agency Innosuisse under the Swiss Competence Center for Energy Research SCCER FEEB$\&$D and the European Research Council under	the ERC Advanced Grant agreement no. 787845 (OCAL). (Corresponding Author: Ahmed Aboudonia)}}
\author{Ahmed Aboudonia, Annika Eichler, Francesco Cordiano, \\ Goran Banjac, and John Lygeros
\footnote{A. Aboudonia, F. Cordiano, G. Banjac and J. Lygeros are with the Automatic Control Laboratory, Department of Electrical Engineering and Information Technology, ETH Zurich, 8092 Zurich, Switzerland. (emails: $\{$ahmedab,gbanjac,lygeros$\}$@control.ee.ethz.ch and fcordiano@student.ethz.ch). A. Eichler is with the Deutsches Elektronen-Synchroton DESY, 22607 Hamburg, Germany (e-mail: annika.eichler@desy.de).}}

\date{}
\maketitle

\begin{abstract}
	
	Various efforts have been devoted to developing stabilizing distributed Model Predictive Control (MPC) schemes for tracking piecewise constant references. In these schemes, terminal sets are usually computed offline and used in the MPC online phase to guarantee recursive feasibility and asymptotic stability. Maximal invariant terminal sets do not necessarily respect the distributed structure of the network, hindering the distributed implementation of the controller. On the other hand, ellipsoidal terminal sets respect the distributed structure, but may lead to conservative schemes. In this paper, a novel distributed MPC scheme is proposed for reference tracking of networked dynamical systems where the terminal ingredients are reconfigured online depending on the closed-loop states to alleviate the aforementioned issues. The resulting non-convex infinite-dimensional problem is approximated using a quadratic program. The proposed scheme is tested in simulation where the proposed MPC problem is solved using distributed optimization.
	
\end{abstract}


\section{Introduction}
\label{sec:introduction}
Various distributed Model Predictive Control (MPC) schemes have been proposed for constrained networked dynamical systems. This is because distributed MPC has several advantages, such as increased privacy, robustness against failure and scalability when controlling such systems (see, e.g. \cite{maestre2014distributed,negenborn2014distributed}). Although many of these schemes are developed for regulation problems, tracking non-zero target points is found to be crucial in many applications. Thus, several distributed MPC schemes have been developed for tracking piecewise constant references (see, e.g. \cite{farina2013solution,razzanelli2017parsimonious,kogel2013set}). In \cite{ferramosca2013cooperative}, a distributed MPC scheme is developed where the maximal invariant terminal set for tracking developed in \cite{limon2008mpc} is used. This terminal set, however, does not respect the distributed structure of the system and couples all subsystems. Polytopic sets can still be used with distributed MPC while respecting the structure \cite{riverso2015plug,boem2019distributed}. A distributed MPC scheme with ellipsoidal terminal sets is also developed in \cite{conte2013cooperative}. Although this scheme respects the distributed structure of the system, it turns out to be conservative leading to relatively small feasible regions.
Various methods have been developed to alleviate the conservatism imposed by terminal sets and enlarge the resulting feasible regions. These methods include using a reference governor \cite{nicotra2018embedding,di2018cascaded}, dynamic terminal set transformation \cite{simon2014reference}, generalized terminal ingredients \cite{mayne2016generalized,fagiano2013generalized} and construction of terminal sets using feasible trajectories \cite{brunner2015stabilizing}. Although the above-mentioned methods are developed for centralized MPC schemes, some are extended to distributed MPC schemes as in \cite{trodden2017distributed,darivianakis2019distributed,aboudonia2019distributed,wang2016distributed} where the terminal ingredients are computed online. However, these schemes are mainly developed for regulation problems.

In this work, we develop a novel distributed MPC with reconfigurable terminal ingredients for reference tracking of networked dynamical systems with a distributed structure. Unlike \cite{ferramosca2013cooperative}, the terminal ingredients are designed to respect the distributed structure while alleviating the conservatism of \cite{conte2013cooperative}. Although the resulting optimal control problem is infinite-dimensional, it can be formulated as a semi-infinite program by restricting the terminal ingredients to ellipsoidal sets and affine controllers. Ellipsoidal terminal sets are used in this work since they can be defined using the level sets of the Lyapunov function. Using robust optimization tools, the infinite number of constraints is then transformed into a finite number of matrix inequalities yielding a finite, albeit non-convex mathematical program. This is in turn shown to be equivalent to a semidefinite program (SDP) through a change of variables. To improve computational performance, the resulting SDP can further be approximated by a quadratic program using diagonal dominance \cite{ahmadi2017sum}. We prove that the proposed scheme is recursively feasible and the target point under this controller is asymptotically stable. Finally, we evaluate the efficacy of the proposed scheme via simulation where we solve the MPC problem using distributed optimization techniques \cite{boyd2011distributed}. 

In Section II, the distributed MPC problem formulation is introduced, followed by the distributed MPC scheme in Section III. The asymptotic stability and recursive feasibility are established in Section IV, while Section V presents the two numerical examples. Section VI provides concluding remarks.
\section{PROBLEM FORMULATION}

\label{sec2}

We consider networked dynamical systems with linear time-invariant dynamics subject to polytopic state and input constraints. We assume that these systems can be decomposed into a set of $M$ subsystems, each of which has a set of neighbours $\pazocal{N}_i$ for all $i \in \{1,\ldots,M\}$. Two subsystems are considered as neighbors if the states of one appear in the dynamics and/or constraints of the other. We assume that $i \in \pazocal{N}_i$ for all $i =\{1, \ldots, M\}$.

We denote the state and input vectors of the $i$-th subsystem by $x_i \in \mathbb{R}^{n_i}$ and $u_i \in \mathbb{R}^{m_i}$, respectively. We also define $x_{N_i} \in \mathbb{R}^{n_{N_i}}$ to be a concatenated state vector comprising the states of the subsystems in the set $\pazocal{N}_i$. Inspired by \cite{ferramosca2013cooperative,conte2013cooperative}, the standard tracking distributed MPC problem is given by
\vspace{-0.75cm}
\begin{table}[H]
	\normalsize
	\begin{subequations}
		\label{sec2_basic}
		\begin{align}
			& \nonumber \min_{x_i(t),u_i(t),x_{e_i},u_{e_i}} \ \sum_{i=1}^M J_i(x_{N_i},u_i,x_{e_{N_i}},u_{e_i}) \ s.t. \\ 
			& \nonumber \forall i \in \{1,\ldots,M\} \ \& \ \forall t \in \{0,\ldots,T-1\} \\		
			& \label{sec2_stdA} x_{N_i}(0) = x_{N_{i,0}}, \quad x_i(t+1) = A_ix_{N_i}(t)+B_i u_i(t), \\
			& x_{N_i}(t) \in \pazocal{X}_{N_i} = \{ x_{N_i} \in \mathbb{R}^{n_{N_i}}: G_i x_{N_i} \leq g_i \}, \\
			& u_i(t) \in \pazocal{U}_i = \{ u_i \in \mathbb{R}^{m_i}: H_i u_i \leq h_i \}, \\	
			& \label{sec2_stdD} x_{e_i} = A_i x_{e_{N_i}} +B_i u_{e_i}, \ x_{e_{N_i}} \in \operatorname{int}(\pazocal{X}_{N_i}), \\
			& \label{sec2_stdD2} u_{e_i} = \kappa_i(x_{e_{N_i}}) \in \operatorname{int}(\pazocal{U}_i), \\
			& \label{sec2_stdE} x_i(T) \in \pazocal{X}_{f_i},
		\end{align}
	\end{subequations}
\end{table}
\vspace{-0.75cm}
\noindent where $T \in \mathbb{N}_+$ is the prediction horizon, $A_i \in \mathbb{R}^{n_i \times n_{N_i}}$ and $B_i \in \mathbb{R}^{n_i \times m_i}$ are the system matrices, $\pazocal{X}_{N_i}$ and $\pazocal{U}_i$ are the state and input constraint sets defined by matrices $G_i \in \mathbb{R}^{n_{q_i} \times n_{N_i}}$, $H_i \in \mathbb{R}^{n_{r_i} \times m_i}$, $g_i \in \mathbb{R}^{n_{q_i}}$ and $h_i \in \mathbb{R}^{n_{r_i}}$, the pair $(x_{e_i},u_{e_i})$ is an artificial equilibirum point to which we aim to converge at the current timestep, $\kappa_i(\cdot)$ is the terminal controller, $\pazocal{X}_{f_i}$ is the terminal set and $\operatorname{int}(\cdot)$ refers to the interior of a set. We also define $x_{e_{N_i}}$ as the vector comprising the artificial equilibrium points of the subsystems in the set $\pazocal{N}_i$ and $x_{N_{i,0}}$ as the current states of the subsystems in the set $\pazocal{N}_i$. We assume that
$
\label{sec21_lc}
J_i(x_{N_i},u_i,x_{e_i},u_{e_i}) = \sum_{t=0}^{T-1} \left\{
\|x_{N_i}(t)-x_{e_{N_i}}\|^2_{Q_i} + \|u_i(t)-u_{e_i}\|^2_{R_i} \right\}
+ \|x_i(T)-x_{e_i}\|^2_{P_i} + \|x_{e_i}-x_{r_i}\|^2_{S_i},
$
where $Q_i \in \mathbb{S}_{++}^{n_{N_i}}$, $R_i \in \mathbb{S}_{++}^{m_i}$, $P_i \in \mathbb{S}_{++}^{n_i}$ and $S_i \in \mathbb{S}_{++}^{n_i}$ and $x_{r_i} \in \mathbb{R}^{n_i}$ is the target point.

Unlike the standard distributed MPC schemes for tracking, we assume that the local terminal controllers $\kappa_i(\cdot)$ for all $i \in \{1,...,M\}$ and the local terminal sets $\pazocal{X}_{f_i}$ for all $i \in \{1,...,M\}$ are decision variables. Note that, in this case, the terminal ingredients depend on the closed-loop states. We restrict the terminal controllers to the set of affine functions (i.e. $\kappa_i(x_{N_i})=K_i x_{N_i} + d_i$ where $K_i \in \mathbb{R}^{m_i \times n_{N_i}}$ is the control gain matrix and $d_i \in \mathbb{R}^{m_i}$ is the feedforward term) and the terminal sets to the set of ellipsoids (i.e. $\pazocal{X}_{f_i}=\{x_i \in \mathbb{R}^{n_i} : (x_i-c_i)^\top P_i (x_i-c_i) \leq \alpha_i\}$ where $c_i$ and $\alpha_i$ determine the center and size of the terminal set, respectively). In this case, the terminal constraint in \eqref{sec2_stdE} can be written by means of Schur complement as
\begin{equation}
	\label{sec2_trmCon}
	\begin{bmatrix}
	P_i^{-1} \alpha_i^{1/2} & x_i(T)-c_i \\
	x_i(T)^\top-c_i^\top & \alpha_i^{1/2}
	\end{bmatrix}
	\geq
	0.
\end{equation}

Since the terminal ingredients are computed online, extra constraints should be added to \eqref{sec2_basic} to ensure asymptotic stability of the terminal dynamics and invariance of the terminal set. For this purpose, we make use of the conditions derived in \cite{jokic2009decentralized,darivianakis2019distributed}. In the case of affine terminal controllers and ellipsoidal terminal sets, these conditions reduce to
\vspace{-0.75cm}
\begin{table}[H]
	\normalsize
	\begin{subequations}
		\label{sec2_trm}
		\begin{gather}
		\nonumber
		\forall i \in \{1,...,M\}, \quad \forall j \in \pazocal{N}_i, \quad \forall x_j \in \pazocal{X}_{f_j}, \\
		\label{sec2_trmA}
		\|(A_i+B_iK_i)x_{N_i}+B_id_i-c_i\|^2_{P_i} \leq \alpha_i, \\
		\label{sec2_trmB}
		G_i x_{N_i} \leq g_i, \\
		\label{sec2_trmC}
		H_i (K_i x_{N_i} + d_i) \leq h_i, \\
		\label{sec2_trmD}
		\delta_{1,i} \|x_i-x_{e_i}\|^2 \leq \|x_i-x_{e_i}\|^2_{P_i} \leq \delta_{2,i} \|x_i-x_{e_i}\|^2, \\
		\label{sec2_trmE}
		\delta_{3,i} \|x_{N_i}-x_{e_{N_i}}\|^2 \leq \|x_{N_i}-x_{e_{N_i}}\|^2_{Q_i+K_i^\top R_i K_i}, \\
		\label{sec2_trmF}
		\|x_i-x_{e_i}\|^2_{P_i} - \|x_{N_i}-x_{e_{N_i}}\|^2_{(A_i+B_i K_i)^\top P_i (A_i + B_i K_i) + Q_i + K_i^\top R_i K_i - \Gamma_i} \geq 0, \\
		\label{sec2_trmG}
		\sum_{i=1}^{M} x_{N_i}^\top {\Gamma_i} x_{N_i} \leq 0.
		\end{gather}
	\end{subequations}
\end{table}
\vspace{-0.75cm}
\noindent where $\delta_{1,i}$, $\delta_{2,i}$ and $\delta_{3,i}$ are positive scalars and $\Gamma_i \in \mathbb{S}^{n_{N_i}}$ are symmetric matrices. Notice that the constraints are uncountable as they should be satisfied for all $x_j$ in the ellipsoidal sets $\pazocal{X}_{f_j}$. Although $\Gamma_i$ are decision variables in the optimal control problem, the scalars $\delta_{1,i}$, $\delta_{2,i}$ and $\delta_{3,i}$ are known a priori and their choice is discussed in Section \ref{sec3}. Constraints \eqref{sec2_trmA}-\eqref{sec2_trmC} are responsible for ensuring positive invariance of the terminal sets, whereas \eqref{sec2_trmD}-\eqref{sec2_trmG} are responsible for ensuring stability of the terminal dynamics. If one uses a global terminal set, \eqref{sec2_trmD}-\eqref{sec2_trmG} provide implicit conditions for the invariance of this terminal set making \eqref{sec2_trmA} redundant. Here, however, we consider local terminal sets, thus \eqref{sec2_trmA} is still required.

Combining \eqref{sec2_basic}, \eqref{sec2_trmCon} and \eqref{sec2_trm} leads to
	\begin{align}
		\label{sec2_ocp}
		& \min_{DV} \ \sum_{i=1}^M J_i(x_{N_i},u_i,x_{e_i},u_{e_i}) \ s.t.\\ 	
		& \nonumber \quad   \eqref{sec2_stdA}-\eqref{sec2_stdD2}, \eqref{sec2_trmCon} \ \forall i \in \{1,\ldots,M\}, \ \forall t \in \{0,\ldots,T-1\}, \\
		& \nonumber \quad  \eqref{sec2_trmA}-\eqref{sec2_trmG} \ \forall i \in \{1,\ldots,M\}, \ \forall j \in \pazocal{N}_i, \  \forall x_j \in \pazocal{X}_{f_j},
	\end{align}
where the decision variables are $DV = \{x_i(t),u_i(t),x_{e_i},u_{e_i},x_i(T),\alpha_i,c_i,K_i,d_i,\Gamma_i\}$ for all $i \in \{1,...,M\}$ and $t \in \{0,...,T-1\}$. Note that the matrix $P$ is the result of an offline optimization problem as proposed in \cite{conte2016distributed}.

For ease of notation, we denote the global state and input vectors of the overall system by $x = (x_1, \ldots, x_M) \in \mathbb{R}^n$ and $u = (u_1, \ldots, u_M) \in \mathbb{R}^m$. Hence, the global dynamics is given by $x(t+1)=Ax(t)+Bu(t)$ where $A \in \mathbb{R}^{n \times n}$ and $B \in \mathbb{R}^{n \times m}$. We assume that the pair $(A,B)$ is controllable. We also denote the global artificial equilibrium and target points by $(x_e,u_e)$ and $x_r=(x_{r_1}, \ldots, x_{r_M}) \in \mathbb{R}^n$ where $x_e=(x_{e_1}, \ldots, x_{e_M}) \in \mathbb{R}^n$ and $u_e=(u_{e_1}, \ldots, u_{e_M}) \in \mathbb{R}^m$. Hence, the global cost function can be written as $J=\sum_{t=0}^{T-1} \left\{ \|x(t)-x_e\|^2_Q + \|u(t)-u_e\|^2_R \right\} + \|x(T)-x_e\|^2_P + \|x_e-x_r\|^2_S$ where $Q \in \mathbb{S}_{++}^{n}$, $R \in \mathbb{S}_{++}^{m}$, $P \in \mathbb{S}_{++}^{n}$ and $S \in \mathbb{S}_{++}^{n}$. Finally, the global terminal set as well as the global state and input constraint sets are denoted by $\pazocal{X}_f$, $\pazocal{X}$ and $\pazocal{U}$, respectively, and the global terminal controller is denoted by $\kappa(x)=Kx+d$. Note that the global terminal set is defined as $\pazocal{X}_f = \pazocal{X}_{f_1} \times \hdots \times \pazocal{X}_{f_M}$. The matrices $A$, $B$, $Q$, $R$, $P$, $S$ and $K$, the vector $d$ and the sets $\pazocal{X}$ and $\pazocal{U}$ can be constructed using the local matrices, vectors and sets in the obvious way. The local variables of the $i$-th subsystem can be extracted from the global variables using the mappings $U_i \in \{0,1\}^{n_i \times n}$, $W_i \in \{0,1\}^{ n_{N_i} \times n}$ and $V_i \in \{0,1\}^{m_i \times m}$ where
\begin{equation}
\begin{aligned}
\label{sec2_map}
x_i = U_i x, \quad x_{N_i} = W_ix, \quad u_i = V_i u.
\end{aligned}
\end{equation}
To ensure that the target point $x_r$ is reachable, we assume that it satisfies the state constraints (and the corresponding input  satisfies the input constraints).

\section{DISTRIBUTED MPC SCHEME}

\label{sec3}

The optimization problem \eqref{sec2_ocp} involves a finite number of decision variables but an infinite number of constraints. This is because constraints \eqref{sec2_trmA}-\eqref{sec2_trmG} should be satisfied for all $x_j \in \pazocal{X}_{f_j}$ where $j \in \pazocal{N}_i$ and $i \in \{1,...,M\}$. We show how these constraints can be transformed into a finite number of matrix inequalities. To simplify the notation, we define $\alpha=diag(\alpha_1 I_{n_1},...,\alpha_M I_{n_M})$, $c=[c_1^\top,...,c_M^\top]^\top$, $\alpha_{{N}_i}=W_i \alpha W_i^\top$ and $c_{{N}_i}=W_i c$. 

We start with \eqref{sec2_trmA} which ensures the invariance of local terminal sets; where \eqref{sec3_intA} and \eqref{sec3_trmA} are shown overleaf in single column.
\begin{prop}
	\label{slemma1}
	The terminal set invariance condition \eqref{sec2_trmA} of the $i$-th subsystem holds for all $x_j \in \pazocal{X}_{f_j}$ where $j \in \pazocal{N}_i$ if there exist scalars $\rho_{ij}\geq0$ such that (\ref{sec3_trmA}) holds where $P_{ij}=W_i U_j^\top P_j U_j W_i^\top$.
\end{prop}
\begin{proof}
	Define the auxiliary vector $s \in \mathbb{R}^n$ such that
	$ x = c + \alpha^{1/2} s $
	where $s_i = U_i s$ and $s_{N_i} = W_i s$. Substituting $s_i$ and $s_{N_i}$ in \eqref{sec2_trmA}, using the mapping equations (\ref{sec2_map}) and multiplying the resulting inequality by $\alpha_i^{-1/2}$ leads to (\ref{sec3_intA}). Applying the S-Lemma \cite{boyd1994linear} followed by the Schur complement results in \eqref{sec3_trmA}.
\end{proof}

\begin{table*}
	\normalsize
	\begin{equation}
	\label{sec3_intA}
	\left.
	\begin{aligned}
	&s_{{N}_i}^\top (A_i \alpha_{{N}_i}^{1/2}+B_i K_i \alpha_{{N}_i}^{1/2})^\top P_i \alpha_i^{-1/2} (A_i \alpha_{{N}_i}^{1/2}+B_i K_i \alpha_{{N}_i}^{1/2}) s_{{N}_i}
	\\ & + 2
	[A_i c_{{N}_i} + B_i ( K_i c_{{N}_i} +d_i ) - c_i]^\top P_i \alpha_i^{-1/2} (A_i \alpha_{{N}_i}^{1/2}+B_i K_i \alpha_{{N}_i}^{1/2}) s_{{N}_i}
	\\ & +
	[A_i c_{{N}_i} + B_i (K_i c_{{N}_i} + d_i) - c_i]^\top P_i \alpha_i^{-1/2}
	[A_i c_{{N}_i} + B_i (K_i c_{{N}_i} + d_i) - c_i]
	\leq \alpha_i^{1/2}
	\end{aligned}
	\right\}
	\forall  j \in \pazocal{N}_i, \ s_{{N}_i}^\top P_{ij} s_{{N}_i} \leq 1
	\end{equation}
	\normalsize
	\begin{equation}
	\label{sec3_trmA}
	\begin{bmatrix}
	P_i^{-1} \alpha_i^{1/2} & (A_i \alpha_{{N}_i}^{1/2}+B_i K_i \alpha_{{N}_i}^{1/2})
	&
	[A_i c_{{N}_i} + B_i (K_i c_{{N}_i} + d_i) - c_i]
	\\
	(A_i \alpha_{{N}_i}^{1/2}+B_i K_i \alpha_{{N}_i}^{1/2})^\top & \sum_{j \in \pazocal{N}_i} \rho_{ij} P_{ij} & 0 \\
	[A_i c_{{N}_i} + B_i (K_i c_{{N}_i} + d_i) - c_i]^\top & 0 & \alpha_i^{1/2} - \sum_{j \in \pazocal{N}_i} \rho_{ij}
	\end{bmatrix}
	\geq 0
	\end{equation}
	\begin{equation}
	\label{sec3_trmF}
	\begin{bmatrix}
	W_i U_i^\top P_i U_i W_i^\top \alpha_i^{1/2} + F_{{N}_i} &  \alpha_{{N}_i}^{1/2} A_i^\top + \alpha_{{N}_i}^{1/2} K_i^\top B_i^\top & \alpha_{{N}_i}^{1/2} Q_i^{{1/2}^\top} & \alpha_{{N}_i}^{1/2} K_i^\top R_i^{1/2} \\
	A_i \alpha_{{N}_i}^{1/2} + B_i K_i \alpha_{{N}_i}^{1/2} & P_i^{-1} \alpha_i^{1/2} & 0 & 0 \\
	Q_i^{1/2} \alpha_{{N}_i}^{1/2} & 0 & \alpha_i^{1/2} I_{n_{N_i}} & 0 \\
	R_i^{1/2}  K_i \alpha_{{N}_i}^{1/2} & 0 & 0 & \alpha_i^{1/2} I_{m_i}
	\end{bmatrix} \geq 0
	\end{equation}
	\rule{\textwidth}{0.4pt}
\end{table*}


Next, we proceed with \eqref{sec2_trmB} which ensures that all state constraints are satisifed inside the local terminal sets.

\begin{prop}
	\label{slemma2}
	Let $G^k_i$ be the $k$-th row of the matrix $G_i$ and $g^k_i$ the $k$-th element of the vector $g_i$. The $k$-th state constraint of the $i$-th subsystem given by
	$
	G_i^k x_{N_i} \leq g_i^k, \ \forall j \in \pazocal{N}_i, \ x_j : (x_j-c_j)^\top P_j (x_j-c_j) \leq \alpha_j
	$
	holds if there exist $\sigma_{ij}^k \geq 0$ such that
	\begin{equation}
	\label{sec3_trmB}
	\begin{aligned}
	\begin{bmatrix}
	\sum_{j\in\pazocal{N}_i}\sigma_{ij}^k P_{ij} & \frac{1}{2} \alpha_{{N}_i}^{1/2} G_i^{k^\top} \\
	\frac{1}{2} G_i^k \alpha_{{N}_i}^{1/2} & g_i^k - 		G_i^k   c_{{N}_i} - \sum_{j\in\pazocal{N}_i}\sigma_{ij}^k
	\end{bmatrix}
	\geq 0.
	\end{aligned}
	\end{equation}
\end{prop}

The proof follows that of Proposition 6 in \cite{aboudonia2019distributed}.
Similarly, we transform constraint \eqref{sec2_trmC} into a matrix inequality as follows.

\begin{prop}
	\label{slemma3}
	Let $H^l_{i}$ be the $l$-th row of the matrix $H_{i}$ and $h^l_{i}$ the $l$-th element of the vector $h_{i}$. The $l$-th input constraint of the $i$-th subsystem given by
	\begin{equation}
	\begin{aligned}
	\label{sec3_intC}
	H_{i}^l (K_i x_{N_i} + d_i) \leq h_{i}^l, \forall j \in \pazocal{N}_i, \ x_j:(x_j-c_j)^\top P_j (x_j-c_j) \leq \alpha_j,
	\end{aligned}
	\end{equation}
	holds if  there exist $\tau_{ij}^l \geq 0$ such that
	\begin{equation}
	\label{sec3_trmC}
	\begin{aligned}
	\begin{bmatrix}
	\sum_{j\in\pazocal{N}_i}\tau_{ij}^l P_{ij} & \frac{1}{2} \alpha_{N_i}^{1/2}
	K_i^\top  H_{i}^{l^\top} \\
	\frac{1}{2} H_{i}^l K_i \alpha_{{N}_i}^{1/2} & h_{i}^l -			H_{i}^l (K_i c_{{N}_i} + d_i) - \sum_{j\in\pazocal{N}_i}\tau_{ij}^l
	\end{bmatrix}
	\geq 0.
	\end{aligned}
	\end{equation}
\end{prop}
\begin{proof}
	Recall the definitions of the auxiliary vectors $s_i$ and $s_{{N}_i}$. Substituting these in (\ref{sec3_intC}) and making use of the mapping equations (\ref{sec2_map}) yield
	$$
	\begin{aligned}
	H_i^l K_i \alpha_{{N}_i}^{1/2} s_{{N}_i}
	+H_i^l K_i  c_{{N}_i} + H_i^l d_i
	\leq h_i^l,
	\ \forall j \in \pazocal{N}_i, \ s_{{N}_i}^\top P_{ij} s_{{N}_i} \leq 1.
	\end{aligned}
	$$
	Applying the S-lemma \cite{boyd1994linear} and rearranging result in (\ref{sec3_trmC}).
\end{proof}

Finally, we convert constraints \eqref{sec2_trmD}-\eqref{sec2_trmG} into a finite number of matrix inequalities; where \eqref{sec3_trmF} is given overleaf in single column.

\begin{prop}
	\label{slemma4}
	The stability constraints \eqref{sec2_trmD}-\eqref{sec2_trmG} hold if \eqref{sec3_trmF} holds for all $i \in \{1,...,M\}$ and there exist block-diagonal matrices $T_i \in \mathbb{S}^{n_{N_i}}$ such that
	\begin{subequations}
		\label{sec3_trmG}
		\begin{gather}
		\label{sec3_trmG1}
		\alpha_{N_i}^{1/2} \Gamma_i \alpha_i^{-1/2} \alpha_{N_i}^{1/2} \leq T_i \quad \forall i \in \{1,...,M\},\\
		\label{sec3_trmG2}
		\sum_{j \in \pazocal{N}_i} U_i W_{j}^\top T_j W_{j} U_i^\top \leq 0 \quad \forall i \in \{1,...,M\}.
		\end{gather}
	\end{subequations}	
\end{prop}
\begin{proof}
	Recall that $P_i \in \mathbb{S}_{++}^{n_i}$ and let $\delta_{1,i}=\lambda_{\min}(P_i)$ and $\delta_{2,i}=\lambda_{\max}(P_i)$ where $\lambda_{\min}(P_i)$ and $\lambda_{\max}(P_i)$ are the minimum and maximum eigenvalues of the matrix $P_i$, respectively. Then, condition \eqref{sec2_trmD} is always satisfied. Similarly, recall that $Q_i \in \mathbb{S}_{++}^{n_{N_i}}$ and $R_i \in \mathbb{S}_{++}^{m_i}$ and
	let $\delta_{3,i}=\lambda_{\min}(Q_i+K_i^\top R_i K_i)$ where $\lambda_{\min}(Q_i+K_i^\top R_i K_i)$ is the minimum eigenvalue of the matrix $Q_i+K_i^\top R_i K_i$. Then, condition \eqref{sec2_trmE} is always satisfied. Hence, conditions \eqref{sec2_trmD} and \eqref{sec2_trmE} can be omitted from the optimal control problem.
	Following \cite{darivianakis2019distributed}, we prove that (\ref{sec3_trmF}) is a sufficient condition for \eqref{sec2_trmF} and
	\begin{equation}
	\label{sec3_intG}
	\sum_{i=1}^{M} W_i^\top \alpha_{{N}_i}^{1/2} \Gamma_i \alpha_i^{-1/2} \alpha_{{N}_i}^{1/2} W_i \leq 0.
	\end{equation}
	is a sufficient condition for \eqref{sec2_trmG}.
	Following \cite{conte2016distributed}, we introduce the block-diagonal matrices $T_i$ by requiring the inequalities in \eqref{sec3_trmG1}. Thus, condition \eqref{sec3_intG} can be ensured by means of the inequalities in \eqref{sec3_trmG2}.
\end{proof}

We note that due to the use of the S-lemma in Propositions \ref{slemma1}-\ref{slemma4}, the derived matrix inequalities are only sufficient conditions for the constraints in \eqref{sec2_trmA}-\eqref{sec2_trmG}.
According to Propositions \ref{slemma1}-\ref{slemma4}, we require that the auxiliary decision variables $\rho_{ij}$, $\sigma_{ij}^k$ and $\tau_{ij}^l$ are non-negative, that is,
\begin{equation}
\label{sec3_slemma}
\begin{aligned}
&\rho_{ij} \ge 0, \ \sigma^k_{ij} \ge 0, \ \tau^l_{ij} \ge 0, \\ 
\forall k \in \{1,...,q_i\},
&\forall l \in \{1,...,r_i\},
\forall j \in \pazocal{N}_i,
\forall i \in \{1,...,M\}.
\end{aligned}
\end{equation}

In summary, the MPC problem is modified to 
\begin{align}
\label{sec3_ocp}
& \min_{DV} \ \sum_{i=1}^M J_i(x_{N_i},u_i,x_{e_{N_i}},u_{e_i}) \ s.t.\\ 	
& \nonumber \quad \quad  \eqref{sec2_stdA}-\eqref{sec2_stdD2}, \ \eqref{sec2_trmCon} \ \forall i \in \{1,\ldots,M\}, \ \forall t \in \{0,\ldots,T-1\}, \\
& \nonumber \quad \quad \eqref{sec3_trmA},\eqref{sec3_trmB},\eqref{sec3_trmC}-\eqref{sec3_trmG},\eqref{sec3_slemma} \ \forall i \in \{1,\ldots,M\}, \ \forall j \in \pazocal{N}_i.
\end{align}

Unlike the optimization problem \eqref{sec2_ocp}, the problem \eqref{sec3_ocp} has a finite number of constraints. Note that \eqref{sec3_ocp} provides an upper bound on the optimal cost in \eqref{sec2_ocp}, as the feasible set has been restricted through using the S-lemma.
The optimization variables in \eqref{sec3_ocp} become $DV = \{x_i(t),u_i(t),x_{e_i},u_{e_i},$ $x_i(T),\alpha_i,c_i,K_i,d_i,\Gamma_i,T_i,\rho_{ij},\sigma_{ij}^k,\tau_{ij}^l\}$ for all $i \in \{1,...,M\}$, $t \in \{0,...,T-1\}$, $j \in \pazocal{N}_i$, $k \in \{1,...,q_i\}$ and
$l \in \{1,...,r_i\}$. Note that although the target point $x_r$ might not be initially included in the terminal set, we aim that it belongs to the interior of the terminal set at steady state.

Problem \eqref{sec3_ocp} is non-convex, due to the nonlinear combinations of decision variables in some of the constraints. It can, however, be transformed into an SDP through the change of variables 
\begin{gather}
\label{sec3_map}
\begin{split}
& v_{1,i} = \alpha_{{N}_i}^{1/2}, \quad v_{3,i} = K_i \alpha_{{N}_i}^{1/2}, \quad v_{5,i} = \alpha_{N_i}^{1/2} \Gamma_i \alpha_i^{-1/2} \alpha_{N_i}^{1/2}, \\
& v_{2,i} = c_{{N}_i}, \quad v_{4,i} = K_i c_{{N}_i} + d_i, \quad v_{6,i} = T_i.
\end{split}
\end{gather}
Equation \eqref{sec3_map} defines a bijective map as long as $\alpha_i > 0 \ \forall i \in\{1,...,M\}$. Hence, the equations in \eqref{sec3_map} do not have to be added to the optimization problem.
As we consider affine terminal controllers, the only non-convex constraint remaining after \eqref{sec3_map} is \eqref{sec2_stdD2}, where the product $K_i x_{e_{N_i}}$ appears. To express this constraint as a linear combination of the variables in \eqref{sec3_map}, the artificial equilibrium point is constrained to be at the center of the terminal set (i.e. $x_{e_i}=c_i$). Using the map \eqref{sec2_map}, this constraint then becomes 
\begin{equation}
	\label{sec3_eq}
	u_{e_i} =  v_{4,i} \in \operatorname{int}(\pazocal{U}_i)
\end{equation}
In the sequel, we denote \eqref{sec3_ocp} with \eqref{sec3_eq} replacing \eqref{sec2_stdD2} and the decision variables in \eqref{sec3_map} replacing the actual decision variables by \textbf{\textit{RTI}} as an abbreviation for distributed MPC with \textbf{\textit{R}}econfigurable \textbf{\textit{T}}erminal \textbf{\textit{I}}ngredients. We note that the derived LMIs are functions of the closed-loop states and hence cannot be solved offline.

\begin{remark}
	Although the MPC problem \textbf{\textit{RTI}} is written centrally, it can be solved online using distributed optimization algorithms (see e.g. \cite{aybat2016distributed}) thanks to its distributed structure.
	Some of these algorithms can be used without requiring a central coordinator such as the distributed primal-dual algorithm \cite{aybat2016distributed} and some variants of ADMM \cite{boyd2011distributed,banjac2019decentralized}. ADMM is used here due to its better convergence properties \cite{banjac2019decentralized}. In this case, each subsystem solves a local optimization problem iteratively while communicating only with its neighbours until consensus among shared variables is reached. The shared variables between two neighbours $i$ and $j$ are $\alpha_i$, $\alpha_j$, $c_i$, $c_j$, $U_j W_i^\top T_i W_i U_j^\top$, $U_i W_j^\top T_j W_j U_i^\top$, $x_i(t)$ and $x_j(t)$ for all $t \in \{0,...,T\}$.
\end{remark}

\begin{remark}
	\label{remark2}
	Unlike standard MPC schemes that require the solution of quadratic programs, the developed scheme yields an SDP that is more difficult to solve. To improve computational performance, the SDP can be approximated, using, for example, diagonal dominance \cite{ahmadi2017sum}. 
\end{remark}

\section{Feasibility and Stability}

In this section, the recursive feasibility of the proposed scheme and the asymptotic stability of the corresponding closed-loop system are established. The proof is inspired from \cite{kouvaritakis2016model,limon2008mpc}. We start by showing the stability of the closed-loop system and that the state and input trajectories converge to the artificial equilibrium trajectory along the lines of \cite{kouvaritakis2016model}. In Lemma \ref{lemma_st1}, we make use of the augmented dynamics $x(t+1)=Ax(t)+Bu(t)$, $x_e(t+1)=x_e(t)$ and $u_e(t+1)=u_e(t)$ and prove that this system has a stable equilibrium point. For this purpose, recall that the pair $(A,B)$ is controllable.

\begin{lemma}
	\label{lemma_st1} 
	The optimal solution of the MPC scheme based on \textbf{\textit{RTI}} is such that $\lim_{k \rightarrow \infty}(x^k(0)-x_e^k)=0$ and $\lim_{k \rightarrow \infty}(u^k(0)-u_e^k)=0$ where $x^k(0)$, $u^k(0)$ and $(x_e^k,u_e^k)$ are, respectively, the first entries of the state and input sequences and the artificial equilibrium corresponding to the optimal solution of \textbf{\textit{RTI}} at time $k$. Moreover, the point $(x,x_e,u_e)=(x_r,x_r,u_r)$ is stable where $u_r$ satisfies the equation $x_r = A x_r + B u_r$.
\end{lemma}

The proof follows the standard MPC stability argument \cite{kouvaritakis2016model} and is omitted in the interest of space. Although the terminal controller is updated at each time instant, the closed-loop dynamics are still time-invariant. In this case, the stability of the target point can still be inferred from Lemma \ref{lemma_st1} using the cost $J^k$ as a Lyapunov function, as $J^k$ is positive definite and $J^{k+1}-J^k$ is negative semidefinite.

Next, we show that, if the optimal state and input trajectories converge to the optimal artificial equilibrium trajectory, then, the optimal artificial equilibrium trajectory converges to the target point. The following proofs are inspired from \cite{limon2008mpc}, however (see lemma statements for the precise definitions), we prove in Lemma \ref{lemma_st2} that $E(x_e^k,\delta_a) \subseteq E(\bar{x}_e,\beta)$ and not just that $x_e^k \in E(\bar{x}_e,\beta)$. Moreover, we prove in Lemma \ref{lemma_st3} that $\|x-\bar{x}_e\|^2_P+\|\bar{x}_e-x_r\|^2_S<\|x_e^k-x_r\|_S^2$ $\forall x \in E(x_e^k,\delta_a)$ instead of $\|x_e^k-\bar{x}_e\|^2_P+\|\bar{x}_e-x_r\|^2_S<\|x_e^k-x_r\|_S^2$. Finally, we prove in Lemma \ref{lemma_st4} that $\lim_{k \rightarrow \infty} (x^k(0)-x_e^k)=0$ implies $\lim_{k \rightarrow \infty} (x_e^k-x_r)=0$ instead of $x^k(0)-x_e^k=0$ implies $ x_e^k-x_r=0$. The proofs of these lemmas are found in the Appendix.

\begin{lemma}
	\label{lemma_st2}
	Let $\kappa^k(x)=K^kx+d^k$ be the terminal controller corresponding to the optimal solution of \textbf{\textit{RTI}}. 
	Then, there exist $\underline{\lambda}(x_e^k)<1$, 
	$\bar{\lambda}(x_e^k)>1$, $\lambda \in [\underline{\lambda}(x_e^k),\bar{\lambda}(x_e^k)]$, 
	$\delta_a>0$ and $\beta>0$ such that the equilibrium point $(\bar{x}_e,\bar{u}_e)=(x_r+\lambda (x_e^k-x_r),u_r+\lambda (u_e^k-u_r))$ satisfies $E(x_e^k,\delta_a) \subset 
	E(\bar{x}_e,\beta)$ where $E(x_e^k,\delta_a)=\{x:\|x-x_e^k\|^2_P \leq \delta_a\}$, $E(\bar{x}_e,\beta)=\{x:\|x-\bar{x}_e\|^2_P \leq \beta\}$ is a positively-invariant set with respect to $\bar{x}_e$ under the controller $\bar{\kappa}(x)=K^kx+\bar{d}$ and $\bar{d}$ satisfies the equation $\bar{x}_e = (A+BK^k)\bar{x}_e + B\bar{d}$. 
\end{lemma}

\begin{lemma}
	\label{lemma_st3}
	Let $\sigma>0$ be such that $S > \sigma P$ and assume that $\delta_1 \in \left(0, \frac{(1+3\sigma) - \sqrt{1+6\sigma+5\sigma^2}}{2}\right)$ and
	$\lambda \in\left(\frac{1+\delta_1-\sqrt{\delta_1^2-(1+3\sigma)\delta_1+\sigma^2}}{1+\sigma}, \frac{1+\delta_1+\sqrt{\delta_1^2-(1+3\sigma)\delta_1+\sigma^2}}{1+\sigma}\right)$. Then, there exist $\delta_a>0$
	such that $\|x-\bar{x}_e\|^2_P+\|\bar{x}_e\|^2_S<\|x_e^k-x_r\|_S^2
	$ for all $x \in E(x_e^k,\delta_a)$
	Moreover, the range in which $\lambda$ can be selected intersects the open set (0,1). 
\end{lemma}

\begin{remark}
	\label{rem1}
	Both Lemma \ref{lemma_st2} and Lemma \ref{lemma_st3} impose constraints on $\lambda$. However, by appropriately choosing the other parameters, these constraints are compatible with each other. In particular, Lemma \ref{lemma_st2} requires $\lambda$ to be in a set containing one in its interior. On the other side, Lemma \ref{lemma_st3} requires $\lambda$ to be in a  range whose upper bound is between zero and one. Note that this upper bound tends to one as $\delta_1$ tends to zero, that is,
	$$
	\lim_{\delta_1 \rightarrow 0} \left( \frac{1+\delta_1}{1+\sigma} + \frac{\sqrt{(1+\delta_1)^2+(1+\sigma)(\sigma-1-3\delta_1)}}{1+\sigma} \right) = 
	1.
	$$
	Thus, if $\delta_1$ is chosen sufficiently small, the upper bound in Lemma \ref{lemma_st3} (which tends to one) can be made higher than the lower bound in Lemma \ref{lemma_st2} (which is lower than one).
\end{remark}

\begin{remark}
	\label{rem2}
	Both Lemma \ref{lemma_st2} and Lemma \ref{lemma_st3} impose constraints on the value of the parameter $\delta_a$. In particular, Lemma \ref{lemma_st2} requires that $E(x_e^k,\delta_a) \subseteq E(\bar{x}_e,\beta)$, whereas Lemma \ref{lemma_st3} requires that $E(x_e^k,\delta_a) \subseteq \pazocal{X}_1$. It is easy to see that, if $\delta_a$ satisfies both conditions, the set $E(x_e^k,\delta_a)$ is positively invariant with respect to the equilibrium $(x_e^k,u_e^k)$ under the terminal controller $\kappa^k(x)=K^kx+d^k$. The invariance condition can be shown in the same way as the invariance condition of the set $E(\bar{x}_e,\beta)$ in Lemma \ref{lemma_st2}. Moreover, the constraint satisfaction condition holds since $E(x_e^k,\delta_a) \subseteq E(\bar{x}_e,\beta) \subseteq E({x}_e^k,\delta_b)$ and for all $x \in E(x_e^k,\delta_b)$, $(x,K^kx+d^k) \in \gamma \pazocal{X} \times \gamma \pazocal{U}$ according to Lemma \ref{lemma_st2}.
\end{remark}

\begin{lemma}
	\label{lemma_st4} 
	If for a given initial state $x_0$, the sequence of optimal solutions to \textbf{\textit{RTI}} is such that $\lim_{k \rightarrow \infty}(x^k(0)-x_e^k)=0$, then $\lim_{k \rightarrow \infty} (x_e^k-x_r)=0$.
\end{lemma}

\begin{theorem}
	\label{th_stability}
	The proposed MPC scheme is recursively feasible and the closed-loop system under this controller is asymptotically stable.
\end{theorem}
\begin{proof}
	Assume that the distributed MPC problem is initially feasible at time $k=0$. Assume that the corresponding optimal predicted state trajectory is $\{x^0(0),x^0(1),...,x^0(T)\}$, the optimal predicted input trajectory is $\{u^0(0),u^0(1),...,u^0(T-1)\}$, the optimal artificial equilibrium $(x_e^0,u_e^0)$, the optimal terminal set parameters are $\{\alpha^0,c^0\}$ and the optimal terminal control parameters are $\{K^0,d^0\}$. Since the optimal terminal set $\pazocal{X}_f$ is designed ensuring the positive invariance properties, then, the state trajectory $\{x^0(1),x^0(2),...,x^0(T),(A+BK^0)x^0(T)+Bd^0\}$, the input trajectory $\{u^0(1),u^0(2),...,u^0(T-1),K^0x^0(T)+d^0\}$, the optimal artificial equilibrium $(x_e^0,u_e^0)$, the optimal terminal set parameters $\{\alpha^0,c^0\}$ and the optimal terminal control parameters $\{K^0,d^0\}$ are a feasible solution to the distributed MPC problem at $k=1$. In other words, the distributed MPC problem is feasible in the next time instant. By induction, the distributed MPC problem is feasible for all $k \geq 1$, or equivalently, recursively feasible.
	
	Lemma \ref{lemma_st1} implies that $\lim_{k \rightarrow \infty}(x^k(0)-x_e^k)=0$ and Lemma \ref{lemma_st4} further implies that $\lim_{k \rightarrow \infty} x_e^k=x_r$ and hence, $\lim_{k \rightarrow \infty}x^k(0)=x_r$. Furthermore, Lemma \ref{lemma_st1} implies that the equilibrium point $x=x_r$ is stable. In conclusion, the origin of the closed-loop system under the proposed distributed MPC scheme is asmptotically stable.
\end{proof}

\section{SIMULATION RESULTS}

The efficacy of the proposed scheme is investigated by means of a benchmark example and an interconnected system example. In both examples, four tracking MPC schemes are compared; centralized MPC (CNT-\hspace{-0.25pt}\cite{limon2008mpc}) where the maximal invariant terminal set is computed offline, distributed MPC (DST-\hspace{-0.25pt}\cite{conte2013cooperative}) where ellipsoidal terminal sets are computed offline, the proposed approach (\textbf{\textit{RTI}}) where the terminal ingredients are computed online and the proposed approach with diagonal dominance (\textbf{\textit{RTI+DD}}) where the terminal ingredients are also computed online. We use the benchmark example to visualize the evolution of the optimal trajectories and terminal ingredients of the proposed approach and to compare the conservatism imposed by the distibuted MPC schemes with respect to the centralized scheme. On the other side, we use the interconnected system example to explore the performance and computational complexity of the proposed approaches. We solve the considered optimization problems using MATLAB with YALMIP \cite{lofberg2004yalmip} and MOSEK \cite{mosek}. Unless otherwise stated, all distributed MPC schemes are solved using ADMM \cite{boyd2011distributed}.

\subsection{Benchmark Example}

\begin{figure*}
	\label{res_fig1}
	\captionsetup[subfigure]{oneside,margin={0.6cm,0cm}}
	\begin{subfigure}{0.3\textwidth}
		\centering
		\includegraphics[scale=0.35]{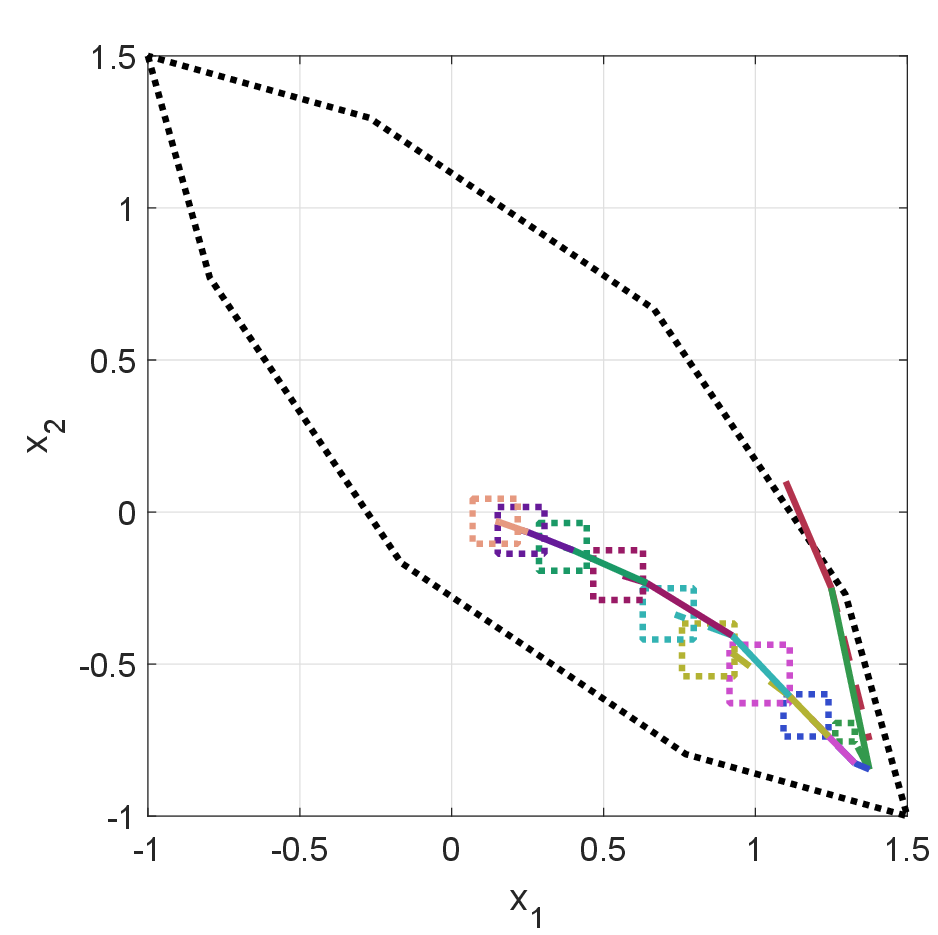}
		\caption{}
		\label{Ex1_Fig1}
	\end{subfigure}
	\hspace{0.25cm}
	\begin{subfigure}{0.3\textwidth}
		\centering
		\includegraphics[scale=0.35]{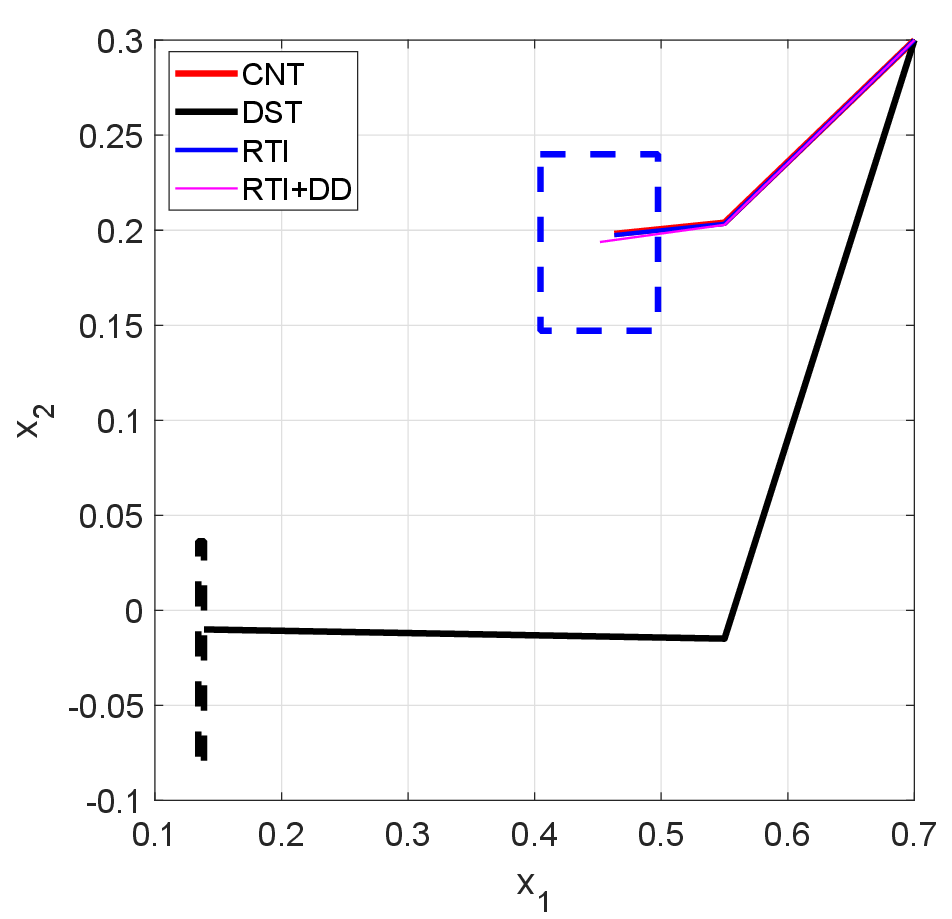}
		\caption{}
		\label{Ex1_Fig2}
	\end{subfigure}
	\hspace{0.3cm}
	\begin{subfigure}{0.3\textwidth}
		\centering 
		\includegraphics[scale=0.35]{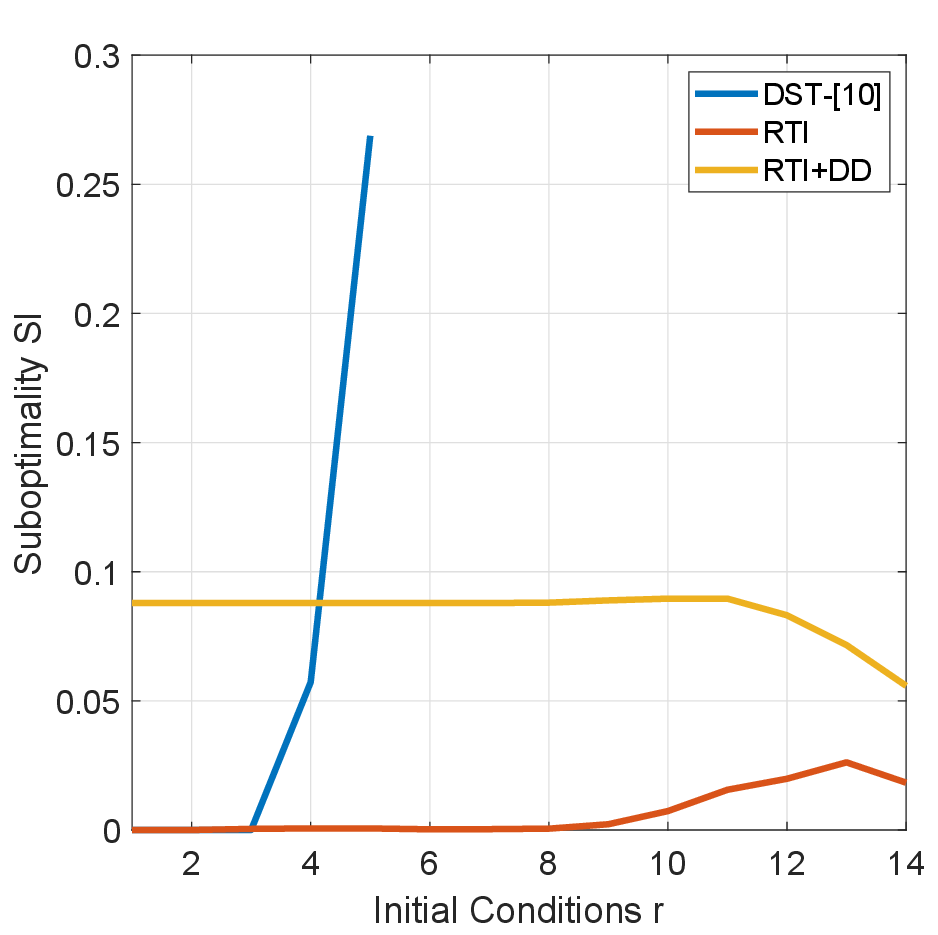}
		\caption{}
		\label{Ex2_Fig}
	\end{subfigure}
	\caption{(a) Benchmark example: evolution of predicted trajectories and terminal sets for \textbf{\textit{RTI}} solved recursively for 10 timesteps (solid: first prediction step, dashed: second prediction step, dotted: terminal set), (b) Benchmark example: comparison of the predicted optimal state trajectories and terminal sets of four tracking MPC schemes, (c) Interconnected system example: Suboptimality indexes of all distributed schemes with respect to the centralized scheme.}
\end{figure*}

The dynamics of the illustrative example is given by $x_1^+=2x_1+0.5x_2-u_1$ and $x_2^+=0.5x_1+2x_2-u_2$, with state and input constraints
$ -5 \leq x_i \leq 5, \ -0.25 \leq u_i \leq 1, \ \text{for } i \in \{1,2\}
$. The system is divided into two neighbouring subsystems with states $x_1$ and $x_2$ and inputs $u_1$ and $u_2$, respectively. The matrices of the cost function are chosen to be $Q_1=Q_2=0.5I_2$, $R_1=R_2=0.1$ and $S_1=S_2=1$, the target point $x_r=[0 \ 0]^\top$ and the prediction horizon $T=2$. The matrix $P$ is computed following \cite{conte2016distributed}.
Fig. \ref{Ex1_Fig1} shows the evolution of the predicted state trajectories and terminal sets of \textbf{\textit{RTI}} solved recursively for 10 timesteps starting from $[x_1 \ x_2]=[1.1 \ 0.1]$. This initial state is outside the maximal invariant terminal set (shown in black) used with CNT-\hspace{-0.25pt}\cite{limon2008mpc}. Note that the optimal state trajectories converge to the target point (i.e. the origin) and the corresponding terminal sets converge to a set containing this target point. Although CNT-\hspace{-0.25pt}\cite{limon2008mpc} and \textbf{\textit{RTI+DD}} yield similar optimal trajectories (omitted in the interest of space), DST-\hspace{-0.25pt}\cite{conte2013cooperative} is found to be initially infeasible starting from this initial condition. This indicates that the feasible region of DST-\hspace{-0.25pt}\cite{conte2013cooperative} is possibly smaller than those of the other three aproaches. Although ellipsoidal terminal sets are utilized, the terminal sets appear as rectangles in Fig. \ref{Ex1_Fig1} since they are the Cartesian products of two one-dimensional ellipsoids.
Fig. \ref{Ex1_Fig2} compares the predicted state trajectories and terminal sets of the four schemes when the optimal control problems are solved once starting from an initial condition $x_1=0.7$ and $x_2=0.3$, that is chosen such that all schemes are initially feasible. 
Although \textbf{\textit{RTI}} and \textbf{\textit{RTI+DD}} lead to very similar predicted trajectories to that of CNT-\hspace{-0.25pt}\cite{limon2008mpc}, DST-\hspace{-0.25pt}\cite{conte2013cooperative} results in a predicted trajectory with higher open-loop cost. This is mainly because the terminal set of DST-\hspace{-0.25pt}\cite{conte2013cooperative} is found to be relatively conservative, i.e. closer to the origin compared to those of  \textbf{\textit{RTI}} and \textbf{\textit{RTI+DD}}.
Note that the terminal set of \textbf{\textit{RTI}} and that of \textbf{\textit{RTI+DD}} (which is very small in Fig. \ref{Ex1_Fig2}) are different since the cost functions of these MPC problems are not strongly convex with respect to the size and center of the terminal set.


\subsection{Interconnected System Example}

\begin{figure}
	\centering
	\scalebox{1}{		\begin{tikzpicture}[thick,scale=0.8, every node/.style={scale=0.8}]
			\node[draw,thick,rectangle,rounded corners=0.1cm,minimum size=.8cm] (pga1) {Agent 1};
			\node[draw,thick,rectangle,rounded corners=0.1cm,minimum size=.8cm, right = 0.8cm of pga1] (pga2) {Agent 2};
			\node[draw,thick,rectangle,rounded corners=0.1cm,minimum size=.8cm, right = 0.8cm of pga1, above = 0.5cm of pga2] (pga3) {Agent 3};
			\node[draw,thick,rectangle,rounded corners=0.1cm,minimum size=.8cm, right = 0.8cm of pga2] (pga5) {Agent 5};
			\node[draw,thick,rectangle,rounded corners=0.1cm,minimum size=.8cm, right = 0.8cm of pga2, above = 0.5cm of pga5] (pga4) {Agent 4};
			\node[draw,thick,rectangle,rounded corners=0.1cm,minimum size=.8cm, right = 0.8cm of pga5] (pga6) {Agent 6};
			\node[draw,thick,rectangle,rounded corners=0.1cm,minimum size=.8cm, right = 0.8cm of pga5, above = 0.5cm of pga6] (pga7) {Agent 7};
			\draw[Stealth-Stealth,thick] (pga1) -- (pga2);
			\draw[Stealth-Stealth,thick] (pga2) -- (pga3);
			\draw[Stealth-Stealth,thick] (pga2) -- (pga5);
			\draw[Stealth-Stealth,thick] (pga3) -- (pga4);
			\draw[Stealth-Stealth,thick] (pga4) -- (pga5);
			\draw[Stealth-Stealth,thick] (pga5) -- (pga6);
			\draw[Stealth-Stealth,thick] (pga5) -- (pga7);
		\end{tikzpicture}}
	\caption{Topology of the considered interconnected system}
	\label{Topology}
\end{figure}
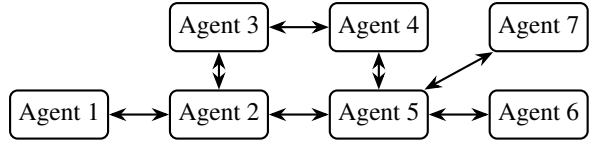

We consider a 7-subsystem interconnected system whose topology is shown in Fig. \ref{Topology}. The dynamics of the $i$-th subsystem (partially adopted from Chapter 2 in \cite{kouvaritakis2016model}) is given by $x_i(t+1) = A_i x_i(t) + B_i u_i(t) + \sum_{j \in \pazocal{N}_i} A_{ij}x_j$ where
$$
A_i = 
\begin{bmatrix}
	1.3 & 2 \\ 0 & 1.15 \\
\end{bmatrix},
\quad
B_i = 
\begin{bmatrix}
	0 \\ 0.0787 \\
\end{bmatrix},
\quad
A_{ij} = 
\begin{bmatrix}
0 & 0.5 \\ 0 & 0 \\
\end{bmatrix}
$$
The $i$-th subsystem is subject to the constraints $[-8 \ -8]^\top \leq x_i \leq [8 \ 8]^\top$ and $-1 \leq u_i \leq 1$. The cost function weights are given by $Q=I_{14}$, $R=10I_7$ and $S=10I_{14}$ where $I_{e}$ is an identity matrix of size $e$. The matrix $P$ is computed offline as in \cite{conte2016distributed}. The origin is chosen to be the target point $x_r$. 

First, we solve the optimal control problem of each scheme $s \in \left\{\right.$CNT-\hspace{-0.25pt}\cite{limon2008mpc}, $\text{DST}$-\hspace{-0.25pt}\cite{conte2013cooperative}, \textbf{\textit{RTI}} and \textbf{\textit{RTI+DD}}$\left.\right\}$ centrally to compare the open-loop cost $J_s^{olc}$ obtained by each scheme when solved to optimality. We use a prediction horizon $T=5$ for all schemes. Fig. \ref{Ex2_Fig} shows the suboptimality index between the centralized and distributed schemes, defined as $SI = (J_{s}^{olc} - J^{olc}_{\text{CNT-\hspace{-0.25pt}\cite{limon2008mpc}}}) / J^{olc}_{\text{CNT-\hspace{-0.25pt}\cite{limon2008mpc}}}$ for the initial conditions $x_{0,i}=[-0.2r \ 0.015r]^\top$ for all $r \in \{1,...,14\}$. The index of  $\text{DST}$-\hspace{-0.25pt}\cite{conte2013cooperative} is shown only for $r \in \{1,...,5\}$ because this scheme is not feasible for the other initial conditions. As the initial condition moves further from the target point, the open-loop cost of $\text{DST}$-\hspace{-0.25pt}\cite{conte2013cooperative} becomes higher than those of \textbf{\textit{RTI}} and \textbf{\textit{RTI+DD}}. This demonstrates the conservatism imposed by $\text{DST}$-\hspace{-0.25pt}\cite{conte2013cooperative} compared to \textbf{\textit{RTI}} and \textbf{\textit{RTI+DD}}. Notice also that \textbf{\textit{RTI+DD}} is more conservative than \textbf{\textit{RTI}} (see Remark \ref{remark2}). Finally, note that the observations based on $J_s^{olc}$ in Fig.\ref{Ex2_Fig} are slightly different from those based on the running cost, which is defined by $J_{s}^{clc}=\sum_{t=1}^{N} \left[ \|x^t(0)-x_r\|_Q^2 + \|u^t(0)-u_r\|_R^2 \right]$ because this cost is different from the cost function in \eqref{sec3_ocp}.

Second, we solve the distributed MPC problems using ADMM \cite{boyd2011distributed} to compare the performance and computational complexity of the distributed schemes. We run the ADMM algorithm for $N=10$ timesteps with the parameter $\rho=1000$ for 100 iterations and denote the optimal running cost of the distributed scheme $s$ obtained using ADMM by $J^{admm}_{s}$; note that $J^{admm}_{s}$ converges to $J_{s}^{clc}$ only asymptotically. We denote the time required by subsystem 5 per timestep to implement ADMM using scheme $s$ by $T^{admm}_{s}$. We choose subsystem 5 as it has the largest number of neigbhours.  Since using longer prediciton horizons is one way of reducing the conservatism imposed by DST-\hspace{-0.25pt}\cite{conte2013cooperative}, we consider two versions of DST-\hspace{-0.25pt}\cite{conte2013cooperative}; $\text{DST}_{5}$-\hspace{-0.25pt}\cite{conte2013cooperative} with $T=5$ and $\text{DST}_{20}$-\hspace{-0.25pt}\cite{conte2013cooperative} with $T=20$. Table \ref{Ex2_Table} compares the distributed MPC schemes in terms of $J^{admm}_{s}$ and $T^{admm}_{s}$ by computing the mean and standard deviation of $J^{dif}_s=|J^{admm}_{s}-J_{s}^{clc}|/J_{s}^{clc}$ and $T^{admm}_{s}$ over all initial conditions for which scheme $s$ is feasible. Despite using longer prediction horizons, $\text{DST}_{20}$-\hspace{-0.25pt}\cite{conte2013cooperative} is still only feasible for $r \in \{1,...,12\}$.
The scheme \textit{\textbf{RTI+DD}} has better convergence properties and smaller computational cost compared to \textit{\textbf{RTI}}, but the latter comes at a fraction of the open-loop cost (see Fig. \ref{Ex2_Fig}).
Although the convergnces properties of $\text{DST}_5$-\hspace{-0.25pt}\cite{conte2013cooperative} are better than those of \textit{\textbf{RTI}}, they are similar to those of \textit{\textbf{RTI+DD}}.
All schemes, however, converge 
faster than $\text{DST}_{20}$-\hspace{-0.25pt}\cite{conte2013cooperative} possibly due to the larger number of shared variables in $\text{DST}_{20}$-\hspace{-0.25pt}\cite{conte2013cooperative}.
The convergence properties of $\text{DST}_{20}$-\hspace{-0.25pt}\cite{conte2013cooperative} could potentially be improved by tuning the ADMM parameters, however $\text{DST}_{5}$-\hspace{-0.25pt}\cite{conte2013cooperative} and $\text{DST}_{20}$-\hspace{-0.25pt}\cite{conte2013cooperative} still yield smaller feasible regions and possibly higher running costs. While the feasible region of $\text{DST}_{20}$-\hspace{-0.25pt}\cite{conte2013cooperative} can be enlarged by further increasing the prediction horizon, this would come at an additional computational cost, which is already higher than \textit{\textbf{RTI+DD}} (though not \textit{\textbf{RTI}}). We note that CNT-\hspace{-0.25pt}\cite{limon2008mpc} requires less time ($\leq 0.1s$ per timestep) than all distributed schemes (Table \ref{Ex2_Table}) due to the ADMM iterations; the distributed schemes, however, generally have other advantages as mentioned in the beginning of Section I.


\section{CONCLUSION}

A novel distributed MPC scheme is proposed for tracking piecewise constant references for interconnected systems. The terminal ingredients are updated online at each time instant. The resulting optimal control problem is approximated using a quadratic program while ensuring recursive feasibility and asymptotic stability. In simulations, the proposed approach has relatively larger feasible regions and stronger scalability properties compared to standard schemes. Ongoing work concentrates on extending this approach to uncertain systems.

\begin{table}
	\normalsize
	\centering
	\caption{Comparison of the four considered distributed MPC schemes in terms of the number of feasible initial conditions $r_{fs}$, the mean $\mu_{c}$ and standard deviation $\sigma_{c}$ of $J_s^{dif}$ and the mean $\mu_{t}$ and standard deviation $\sigma_{t}$ of $T^{admm}_{s}$}
	\begin{tabular}{|c|c|c|c|c|c|}
		\hline
		& $r_{fs}$ & $\mu_{c}$ & $\sigma_{c}$ & $\mu_{t}$ & $\sigma_{t}$ \\
		\hline
		$\text{DST}_5$-\hspace{-0.25pt}\cite{conte2013cooperative} & 5 & 0.0047  &  0.0031  &  0.2539  &  0.0063 \\
		\hline
		$\text{DST}_{20}$-\hspace{-0.25pt}\cite{conte2013cooperative} & 12 & 0.0385  &  0.0087  &  0.5003  &  0.0258 \\
		\hline
		\textbf{\textit{RTI}} & 14 & 0.0089  &  0.0058  &  2.8309  &  0.0962 \\
		\hline
		\textbf{\textit{RTI+DD}} & 14 & 0.0050  &  0.0029  &  0.4455  &  0.0132 \\
		\hline
	\end{tabular}
	\label{Ex2_Table}
	\vspace{-0.5cm}
\end{table}

\section*{APPENDIX}

\begin{refproof}[Proof of Lemma \ref{lemma_st2}:]
This proof is similar to that of Lemma 1 in \cite{limon2008mpc}.
Since $(x_e^k,u_e^k)$ is the artificial equilibrium corresponding to the optimal solution of \textbf{\textit{RTI}}, then $(x_e^k,u_e^k) \in \interior(\pazocal{X} \times \pazocal{U})$.
Define $d_{min}\in[0,1)$ as the smallest scalar such that $(x_e^k,u_e^k) \in d_{min}\pazocal{X} \times d_{min}\pazocal{U}$ and let $\gamma \in (d_{min},1)$. 
Note also that $u_e^k=K^k x_e^k +d^k$ since $\kappa^k(x)=K^kx+d^k$ is the terminal controller corresponding to the optimal solution of \textbf{\textit{RTI}}.
Hence, there exists $\delta_b > 0$ such that for all $x \in E(x_e^k,\delta_b)$, $(x,K^kx+d^k) \in \gamma \pazocal{X} \times \gamma \pazocal{U}$. 
Since $\bar{u}_e=K^k\bar{x}_e+\bar{d}$, then $\bar{d}-d^k=-(1-\lambda)d^k$ as $\bar{x}_e=\lambda x_e^k$ and $\bar{u}_e=\lambda u_e^k$. 
Choose $\lambda \in (\underline{\lambda}(x_e^k),\bar{\lambda}(x_e^k))$ such that $\bar{x}_e = x_r+\lambda (x_e^k-x_r) \in \interior(E(x_e^k,0.25\delta_b))$ and $(0,-(1-\lambda)d^k) \in (1-\gamma)\pazocal{X} \times (1-\gamma)\pazocal{U}$ where $\underline{\lambda}(x_e^k)$ and $\bar{\lambda}(x_e^k)$ are the minimum and maximum values satisfying these inequalities. It is easy to verify, through the last two conditions, that $\underline{\lambda}(x_e^k)<1$ and $\bar{\lambda}(x_e^k)>1$.
Hence, there exists $\beta>0$ such that $x_e^k \in \interior(E(\bar{x}_e,\beta))$ and $E(\bar{x}_e,\beta) \subset E(x_e^k,\delta_b)$. Therefore, there exists $\delta_a>0$ such that $E(x_e^k,\delta_a) \subset E(\bar{x}_e,\beta) \subset E(x_e^k,\delta_b)$.
It remains to prove that $E(\bar{x}_e,\beta)$ is a positively invariant set with respect to $(\bar{x}_e,\bar{u}_e)$ under the controller $\bar{\kappa}(x)=K^kx+\overline{d}$. For all $x \in E(\bar{x}_e,\beta)$, $\|(x-\bar{x}_e)\|^2_P \leq \beta$ and hence, $\|(x-\bar{x}_e)\|_P - \|(x-\bar{x}_e)\|^2_{Q+K^{k^\top} R K^k} \leq \beta$ since $Q>0$ and $R>0$. Thus, $\|(x-\bar{x}_e)\|^2_{(P-Q-{K^k}^\top R {K^k})} \leq \beta$. It is easy to verify from \eqref{sec2_trmF} and \eqref{sec2_trmG} that the matrix $P$ satisfies the Lyapunov inequality $P \geq (A+BK^k)^\top P (A+BK^k) + Q +{K^k}^\top R K^k$.
Thus, $\|(x-\bar{x}_e)\|^2_{(A+BK^k)^\top P (A+BK^k)} \leq \beta$, or equivalently, $\|(x^+-\bar{x}_e)\|^2_P \leq \beta$. In addition, $(x,K^kx+\bar{d})=(x,K^kx+d^k)+(0,\bar{d}-d^k)=(x,K^kx+d^k)+(0,-(1-\lambda)d^k) \in (\gamma \pazocal{X} \times \gamma \pazocal{U}) + ((1-\gamma) \pazocal{X} \times (1-\gamma) \pazocal{U}) =\pazocal{X} \times \pazocal{U}$.
\end{refproof}

\begin{refproof}[Proof of Lemma \ref{lemma_st3}:]
Note that $\|x_e^k-\bar{x}_e\|^2_P = (1-\lambda)^2 \|x_e^k-x_r\|^2_P$ since $x_e^k-\bar{x}_e= (1-\lambda) (x_e^k-x_r)$. 
Notice also that $x-\bar{x}_e = (x-x_e^k)+(1-\lambda)(x_e^k-x_r) $. Hence,
$
\|x-\bar{x}_e\|^2_P
=
\|x-x_e^k\|^2_P 
+ 2 (1-\lambda){(x_e^k-x_r)}^\top P (x-{x_e^k})
+(1-\lambda)^2 \|x_e^k-x_r\|^2_P.
$
Consider a constant $\delta_1 > 0$ and the set $ \pazocal{X}_1$  defined as
$
\pazocal{X}_1 = \left\{x \in \mathbb{R}^{n} : 
\|x-{x}_e^k\|^2_P \leq \delta_1 \|{x}_e^k-x_r\|^2_P, \
\  |{(x_e^k-x_r)}^\top P (x-{x}_e^k)| \leq \right.$ $\left.\delta_1 \|{x}_e^k-x_r\|^2_P 
\right\}.
$ 
For every $\delta_1$ we can select $\delta_a$ small enough such that $E(x_e^k,\delta_a) \subset \pazocal{X}_1$. Therefore, for all $x \in E(x_e^k,\delta_a)$,
$
\|x-\bar{x}_e\|^2_P \leq 
\|x_e^k-x_r\|^2_{(\delta_1 P + 2 \delta_1 (1-\lambda)P + (1-\lambda)^2 P)}.
$
Since $\bar{x}_e-x_r=\lambda (x_e^k-x_r)$, then $\|\bar{x}_e-x_r\|^2_S = \lambda^2 \|x_e^k-x_r\|^2_S$. Hence,
$
\|x-\bar{x}_e\|^2_ P + \|\bar{x}_e-x_r\|^2_S \leq  
\|x_e^k-x_r\|^2_{(\delta_1 P + 2 \delta_1 (1-\lambda)P + (1-\lambda)^2 P + \lambda^2 S)}.
$
To prove that $\|x-\bar{x}_e\|^2_P+\|\bar{x}_e-x_r\|^2_S<\|x_e^k-x_r\|_S^2$, it is required to find conditions on $\delta_1$ and $\lambda$ so that 
$
(1-\lambda^2)S-(1-\lambda)^2 P -2 \delta_1 (1-\lambda) P - \delta_1 P > 0.
$
Since $S>\sigma P$, it suffices to ensure that
$
(1-\lambda^2)\sigma-(1-\lambda)^2 -2 \delta_1 (1-\lambda) - \delta_1 > 0
$
or, equivalently,
$
-(1+\sigma)\lambda^2+2(1+\delta_1)\lambda+(\sigma-1-3\delta_1)>0.
$
Since the quadartic is concave in $\lambda$, its roots are required to be real and distinct so that there exists $\lambda$ which satisfies the strict inequality. The roots are
\begin{equation}
\label{john_eq}
\frac{1+\delta_1}{1+\sigma} \pm \frac{\sqrt{(1+\delta_1)^2+(1+\sigma)(\sigma-1-3\delta_1)}}{1+\sigma}
\end{equation}
and are real and distinct as long as
$
(1+\delta_1)^2+(1+\sigma)(\sigma-1-3\delta_1)>0,
$
or, equivalently,
$
\delta_1^2-(1+3\sigma)\delta_1+\sigma^2>0.
$
This in turn is a convex quadratic in $\delta_1$ whose roots
$
\frac{(1+3\sigma) \pm \sqrt{(1+3\sigma)^2-4\sigma^2}}{2}
$
are real, distinct and positive since
$
(1+3\sigma)^2-4\sigma^2 = 1+6\sigma+5\sigma^2 > 0
$
as $\sigma>0$. If we then pick $\delta_1 \in \left(0, \frac{(1+3\sigma) - \sqrt{1+6\sigma+5\sigma^2}}{2}\right)$, the roots of \eqref{john_eq} are real and distinct. Thus, for any $\sigma>0$, there exists a small enough $\delta_1$ such that there exists $\lambda$ which satisfies the desired condition.
It remains to show that $\lambda$ can be selected in the interval (0, 1). For this, it suffices to prove that it is always possible to choose at least one of the roots in \eqref{john_eq} to be between zero and one. Consider the larger root $\left({1+\delta_1} +\sqrt{(1+\delta_1)^2+(1+\sigma)(\sigma-1-3\delta_1)}\right)/(1+\sigma)$ in \eqref{john_eq}. Note that this root is always positive. For this root to be smaller than or equal to 1, it is required that $1+\sigma \geq {1+\delta_1} +\sqrt{(1+\delta_1)^2+(1+\sigma)(\sigma-1-3\delta_1)}$. Notice that this inequality holds only if $\sigma > \delta_1$. Simplifying and squaring the desired inequality reduces to $\delta_1 (\sigma+1)\geq 0$, which is always the case since $\sigma$ and $\delta_1$ are positive constants. 
In conclusion, for any positive $\delta_1 < \sigma$ such that $\sigma P \leq S$ and $\delta_1 < \frac{(1+3\sigma) - \sqrt{1+6\sigma+5\sigma^2}}{2}$, there exists $\lambda \in (0,1)$ such that
$
(1-\lambda^2)\sigma-(1-\lambda)^2 -2 \delta_1 (1-\lambda) - \delta_1 > 0,
$
and consequently the condition
$
\|x-\bar{x}_e\|^2_P+\|\bar{x}_e-x_r\|^2_S<\|x_e^k-x_r\|_S^2
$
is satisfied.
\end{refproof}

\begin{refproof}[Proof of Lemma \ref{lemma_st4}:]
Assume, for the sake of contradiction, that $\lim_{k \rightarrow \infty}(x^k(0)-x_e^k)=0$ but the sequence of optimal equilibrium points $\{x_e^k\}_{k=0}^\infty$ either does not converge, or does but its limit is not the target point $x_r$. 
In both cases, there exists $\delta_c>0$ such that $\|x_e^k-x_r\|^2_P \geq \delta_c$ for infinitely many $k$. 
Since $\lim_{k \rightarrow \infty}(x^k(0)-x_e^k)=0$, it is always possible to pick an arbitrarily large $k$ such that $x^k(0) \in E(x_e^k,\delta_a)$ where $\|x_e^k-x_r\|^2 \geq \delta_c$ and $\delta_a$ satisfies the conditions in Lemma \ref{lemma_st2} $\left(E(x_e^k,\delta_a) \subseteq E(\bar{x}_e,\beta)\right)$ and Lemma \ref{lemma_st3} $\left(E(x_e^k,\delta_a) \subseteq \pazocal{X}_1\right)$. 
According to Remark \ref{rem2}, it is always the case that the selected $\delta_a>0$ makes the set $E(x_e^k,\delta_a)$ positively invariant with respect to $x_e^k$ under the optimal controller $\kappa^k(x)=K^kx+d^k$. Since $x^k(0) \in E(x_e^k,\delta_a)$, the optimal cost $J^k$ is given by
$
{J}^k = \sum_{t=0}^{T-1} \{ \|{x}^{k}(i)-x_e^k\|^2_Q + \|{x}^{k}(i)-x_e^k\|^2_{{K^k}^\top R {K^k}}	\} 	+ \|{x}^{k}(T)-x_e^k\|^2_P + \|x_e^k-x_r\|^2_S.
$
According to Lemma \ref{lemma_st2}, $E({x}_e^k,\delta_a) \subseteq E(\bar{x}_e,\beta)$ which is a positively invariant set with respect to $\bar{x}_e$ under the terminal controller $\bar{\kappa}(x)=K^kx+\bar{d}$. Thus, there exists a feasible solution starting from the initial condition $\bar{x}(0) = x^k(0) \in E(x_e^k,\delta_a)$ aiming to converge to the non-optimal equilibrum point $\bar{x}_e$. Denote the cost of this feasible solution as
$
\bar{J} = \sum_{t=0}^{T-1} \{
\|\bar{x}(i)-\bar{x}_e\|^2_Q + \|\bar{x}(i)-\bar{x}_e\|^2_{{K}^{k^\top} R {K}^k}	\}
+ \|\bar{x}(T)-\bar{x}_e\|^2_P + \|\bar{x}_e-x_r\|^2_S.
$
Note that $J^k < \bar{J}$ since $J^k$ is the optimal cost. It is easy to verify from \eqref{sec2_trmF} and \eqref{sec2_trmG} that the matrix $P$ satisfies the Lyapunov inequality
$
P \geq (A+BK^k)^\top P (A+BK^k) + Q +{K^k}^\top R K^k
$ and hence that
$
\begin{aligned}
\bar{J} \leq \|{x}^k(0)-\bar{x}_e\|^2_P + \|\bar{x}_e-x_r\|^2_S.
\end{aligned}
$
According to Lemma \ref{lemma_st3}, $\|x-\bar{x}_e\|^2_P+\|\bar{x}_e-x_r\|^2_S<\|{x}_e^k-x_r\|_S^2
$ for all $ x \in E({x}_e^k,\delta_a)$. Since $x^k(0) \in E({x}_e^k,\delta_a)$, then $\|x^k(0)-\bar{x}_e\|^2_P+\|\bar{x}_e-x_r\|^2_S<\|{x}_e^k-x_r\|_S^2
$. Note that $\|{x}_e^k-x_r\|_S^2 \leq J^k$ which contradicts the optimality of $J^k$.
\end{refproof}

\section*{ACKNOWLEDGMENT}

The authors would like to thank Prof. Roy Smith and Dr. Georgios Darivianakis for the fruitful discussions on the topic.

\bibliographystyle{unsrt}
\bibliography{references}

\end{document}